\documentclass[journal,draftcls,onecolumn,12pt,twoside]{IEEEtran}
\usepackage{amsfonts,color,morefloats,pslatex}
\usepackage{amssymb,amsthm, amsmath,latexsym}

\newtheorem{theorem}{Theorem}
\newtheorem{lemma}[theorem]{Lemma}
\newtheorem{proposition}[theorem]{Proposition}
\newtheorem{corollary}[theorem]{Corollary}

\newtheorem{example}[theorem]{Example}

\newcommand{\ord}{{\mathrm{ord}}}

\newcommand{\lcm}{{\mathrm{lcm}}}
\newcommand{\tr}{{\mathrm{Tr}}}

\newcommand{\gf}{{\mathrm{GF}}}
\newcommand{\PG}{{\mathrm{PG}}}

\newcommand{\PAut}{{\mathrm{PAut}}}
\newcommand{\MAut}{{\mathrm{MAut}}}
\newcommand{\GAut}{{\mathrm{Aut}}}

\newcommand{\Sym}{{\mathrm{Sym}}}

\newcommand{\Z}{\mathbb{{Z}}}

\newcommand{\m}{\mathbb{M}}

\newcommand{\cP}{{\mathcal{P}}}
\newcommand{\cB}{{\mathcal{B}}}

\newcommand{\C}{{\mathcal{C}}}

\newcommand{\bc}{{\mathbf{c}}}

\newcommand{\bzero}{{\mathbf{0}}}

\newcommand{\bD}{{\mathbb{D}}}

\newcommand{\PSL}{{\mathrm{PSL}}}
\newcommand{\PGL}{{\mathrm{PGL}}}

\usepackage{blindtext}

\ifCLASSINFOpdf

\else

\fi

\begin{document}
\title{The subfield codes and subfield subcodes of a family of MDS codes\thanks{
C. Tang was supported by the National Natural Science Foundation of China (Grant No.
11871058) and China West Normal University (14E013, CXTD2014-4 and the Meritocracy Research
Funds). Q. Wang was supported by the National Natural Science Foundation of China (Grant No.
11931005). C. Ding's research was supported by the Hong Kong Research Grants Council,
Proj. No. 16300919.}
}

\author{Chunming Tang, \and Qi Wang, \and Cunsheng Ding 

\thanks{C. Tang is with the School of Mathematics  and Information,
China West Normal University, Nanchong 637002,  China (email:
tangchunmingmath@163.com).
}

\thanks{Q. Wang is with the Department of Computer Science and Engineering, and is also with National Center for Applied Mathematics (Shenzhen), Southern University of Science and Technology, Shenzhen 518055, China (email: wangqi@sustech.edu.cn).} 

\thanks{C. Ding is with the Department of Computer Science and Engineering,
The Hong Kong University of Science and Technology, Clear Water Bay,
Kowloon, Hong Kong, China (email: cding@ust.hk).} 

}

\maketitle

\begin{abstract}
Maximum distance separable (MDS) codes are very important in both theory and practice. There is a classical construction of  a family of $[2^m+1, 2u-1, 2^m-2u+3]$ MDS codes for $1 \leq u \leq 2^{m-1}$, which are cyclic, reversible and BCH codes over $\gf(2^m)$. The objective of this paper is to study the quaternary subfield subcodes and quaternary subfield codes of a subfamily of the MDS codes for even $m$.  A family of quaternary cyclic codes is obtained. These quaternary codes are
distance-optimal in some cases and very good in general. Furthermore, infinite families of $3$-designs from these quaternary codes are presented.     
\end{abstract}


% \fi
\IEEEpeerreviewmaketitle

\section{Introduction} 

An $[n, k, d]$ code over $\gf(q)$ is a $k$-dimensional linear subspace of $\gf(q)^n$ 
with minimum Hamming distance $d$.  By the parameters of a linear code, we refer to its length $n$, 
dimension $k$ and minimum distance $d$. 
An $[n, k, d]$ 
code over $\gf(q)$ is called \textit{distance-optimal} (resp. 
\textit{dimension-optimal} and \textit{length-optimal}) if there is no $[n, k, d' \geq d+1]$ (resp. 
$[n, k' \geq k+1, d]$ and $[n' \leq n-1, k, d]$) linear code over $\gf(q)$. An optimal code 
is a code that is length-optimal, or dimension-optimal, or distance-optimal, or meets a bound for linear codes. 

Let $\C$ be an $[n, k, d]$ linear code over $\gf(q)$. Let $A_i$ denote the
number of codewords with Hamming weight $i$ in $\C$ for $0 \leq i \leq n$. The sequence
$(A_0, A_1, \cdots, A_{n})$ of integers is
called the \textit{weight distribution} of $\C$, and the polynomial $\sum_{i=0}^n A_iz^i$ is referred to as
the \textit{weight enumerator} of $\C$. In this paper, $\C^\perp$ denotes the dual code 
of a linear code $\C$,  and $(A_0^\perp, A_1^\perp, \cdots, A_{n}^\perp)$ 
denotes the weight distribution of $\C^\perp$.  

An $[n, k, n-k+1]$ linear code is called a maximum distance separable (MDS) code. 
MDS codes are very important in theory and have important applications. For example, the Reed-Solomon codes 
are widely used in communication systems and data storage devices.  
There is a classical construction of  a family of 
$[2^m+1, 2u-1, 2^m-2u+3]$ MDS codes for $1 \leq u \leq 2^{m-1}$, which are cyclic, reversible and BCH codes 
over $\gf(2^m)$. 
The objective of this paper is to study the quaternary subfield subcodes and quaternary subfield codes 
of a subfamily of the MDS codes for even $m$.  A family of quaternary cyclic codes is obtained. These quaternary codes 
are distance-optimal in some cases and very good in general. Infinite families of 
$3$-designs from these quaternary codes is presented.

\section{Cyclic codes and BCH codes over finite fields} 

Let $q$ be a prime power. An $[n,k, d]$ code $\C$ over $\gf(q)$ is said to be {\em cyclic} if 
$(c_0,c_1, \cdots, c_{n-1}) \in \C$ implies $(c_{n-1}, c_0, c_1, \cdots, c_{n-2}) 
\in \C$.  
We identify any vector $(c_0,c_1, \cdots, c_{n-1}) \in \gf(q)^n$ 
with  
$$ 
c_0+c_1x+c_2x^2+ \cdots + c_{n-1}x^{n-1} \in \gf(q)[x]/(x^n-1).  
$$
Then a code $\C$ of length $n$ over $\gf(q)$ corresponds to a subset of the ring 
$\gf(q)[x]/(x^n-1)$. 
A linear code $\C$ is cyclic if and only if its corresponding subset in $\gf(q)[x]/(x^n-1)$ 
is an ideal of the ring $\gf(q)[x]/(x^n-1)$. 
A cyclic code $\C$ is said to be {\em reversible} if $\C \cap \C^\perp = \{\bzero\}$. Such a code is called a {\em linear complementary dual} (LCD) code. Note that the family of MDS codes presented in the next section are cyclic BCH codes and reversible. 

It is well known that every ideal of $\gf(q)[x]/(x^n-1)$ is principal. Let $\C=\langle g(x) \rangle$ be a 
cyclic code, where $g(x)$ is monic and has the smallest degree among all the 
generators of $\C$. Then $g(x)$ is unique and called the {\em generator polynomial,} 
and $h(x)=(x^n-1)/g(x)$ is referred to as the {\em check polynomial} of $\C$.

Let $n$ be a positive integer and 
let $\Z_n$ denote  the set $\{0,1,2, \cdots, n-1\}$.  Assume that $\gcd(n, q)=1$. For any integer $s$ with $0 \leq s <n$, the \emph{$q$-cyclotomic coset of $s$ modulo $n$\index{$q$-cyclotomic coset modulo $n$}} is defined by 
$$ 
C_s=\{s, sq, sq^2, \cdots, sq^{\ell_s-1}\} \bmod n \subseteq \Z_n,  
$$
where $\ell_s$ is the smallest positive integer such that $s \equiv s q^{\ell_s} \pmod{n}$, and is the size of the 
$q$-cyclotomic coset. The smallest integer in $C_s$ is called the \emph{coset leader\index{coset leader}} of $C_s$. 
Let $\Gamma_{(n,q)}$ be the set of all the coset leaders. It is easily seen then $C_s \cap C_t = \emptyset$ for any two 
distinct elements $s$ and $t$ in  $\Gamma_{(n,q)}$, and  
\begin{eqnarray}\label{eqn-cosetPP}
\bigcup_{s \in  \Gamma_{(n,q)} } C_s = \Z_n. 
\end{eqnarray}
Hence, the distinct $q$-cyclotomic cosets modulo $n$ partition $\Z_n$. 

Let $m=\ord_{n}(q)$ be the order of $q$ modulo $n$, and let $\alpha$ be a generator of $\gf(q^m)^*$. Put $\beta=\alpha^{(q^m-1)/n}$. 
Then $\beta$ is a primitive $n$-th root of unity in $\gf(q^m)$. The minimal polynomial $\m_{\beta^s}(x)$ 
of $\beta^s$ over $\gf(q)$ is the monic polynomial of the smallest degree over $\gf(q)$ with $\beta^s$ 
as a root.  It is obvious that this polynomial is given by 
\begin{eqnarray*}
\m_{\beta^s}(x)=\prod_{i \in C_s} (x-\beta^i) \in \gf(q)[x], 
\end{eqnarray*} 
which is irreducible over $\gf(q)$. It then follows from (\ref{eqn-cosetPP}) that 
\begin{eqnarray*}
x^n-1=\prod_{s \in  \Gamma_{(n,q)}} \m_{\beta^s}(x) ,
\end{eqnarray*}
which is the factorization of $x^n-1$ into irreducible factors over $\gf(q)$.

Let $\delta$ be an integer with $2 \leq \delta \leq n$ and let $h$ be an integer.  
A \emph{BCH code} over $\gf(q)$ 
with length $n$ and \emph{designed distance} $\delta$, denoted by $\C_{(q,n,\delta,h)}$, is a cyclic code with 
the generator polynomial 
\begin{eqnarray}\label{eqn-BCHdefiningSet}
g_{(q,n,\delta,h)}=\lcm(\m_{\beta^h}(x), \m_{\beta^{h+1}}(x), \cdots, \m_{\beta^{h+\delta-2}}(x)) 
\end{eqnarray}
where the least common multiple is computed over $\gf(q)$. 

If $h=1$, the code $\C_{(q,n,\delta,h)}$ with the generator polynomial in (\ref{eqn-BCHdefiningSet}) is called a \emph{narrow-sense} BCH code. If $n=q^m-1$, then $\C_{(q,n,\delta,h)}$ is referred to as a \emph{primitive} BCH code. 

BCH codes are a subclass of cyclic codes with interesting properties and applications. In many cases BCH codes are 
the best linear codes. For example, among all binary cyclic codes of odd length $n$ with $n \leq 125$ the best cyclic 
code is always a BCH code except for two special cases~\cite{Dingbook15}. Reed-Solomon codes are also BCH codes 
and are widely used in communication devices and consumer electronics.  
In the past decade, a lot of progress on the study of BCH codes has been made (see, for example, \cite{LWL19,LiSIAM,LLFLR,SYW,YLLY}). A 70-year breakthrough and a 71-year breakthrough in 
combinatorics were recently made with the help of special families of BCH codes \cite{DingTang19,TangDing20}.   

In this paper, we will use a family of BCH codes of length $2^m+1$ over $\gf(2^m)$ to obtain a family of quaternary 
cyclic codes with interesting parameters and application in combinatorics.

\section{A family of $[2^m+1, 2u-1, 2^m-2u+3]$ MDS codes over $\gf(q)$}\label{sec-LCDMDScycliccode} 

In the following, let $q=2^m$ and $m$ be a positive integer. Let $\alpha$ be a primitive element of $\gf(q^2)$. 
Define $\beta=\alpha^{q-1}$. Then $\beta$ is a $(q+1)$-th primitive root of unity in $\gf(q^2)$. Let $\m_{\beta^i}(x)$ denote 
the minimal polynomial of $\beta^i$ over $\gf(q)$. Clearly, $\m_{\beta^0}(x)=x-1$ and 
$$ 
\m_{\beta^i}(x)=(x-\beta^i)(x-\beta^{q+1-i})=(x-\beta^i)(x-\beta^{-i}) 
$$  
for all $1 \leq i \leq q$. 

For each $u$ with $1 \leq u \leq q/2$, define 
\begin{eqnarray*}
g_u(x)=\m_{\beta^u}(x) \cdots \m_{\beta^{q/2}}(x).  
\end{eqnarray*}
Let $\C_u$ be the cyclic code of length $q+1$ over $\gf(q)$ with the generator polynomial $g_u(x)$. 

The following theorem is well known \cite{MS77}. But for completeness, we give a proof here, which is very short. 

\begin{theorem}\label{thm-sdjoin1}
Let $m \geq 2$ be an integer. For each $u$ with $1 \leq u \leq q/2$, $\C_u$ is a $[q+1, 2u-1, q-2u+3]$ MDS cyclic code 
and is reversible. 
\end{theorem} 

\begin{proof}
By definition, the check polynomial of $\C_u$ is 
$$ 
h_u(x)=\m_{\beta^0}(x) \cdots \m_{\beta^{u-1}}(x). 
$$
For each $i$ with $1 \leq i \leq q/2$, $\deg(\m_{\beta^i}(x))=2$. Consequently, $\deg(h_u(x))=2(u-1)+1=2u-1$. Hence, 
the dimension of $\C_u$, denoted by $\dim(\C_u)$, is $2u-1$. 

By definition, $g_u(x)$ has the following consecutive roots 
$$ 
\beta^u, \beta^{u+1}, \ldots, \beta^{q-u}, \beta^{q-u+1}. 
$$
Note that $(q-u+1)-u+1=q-2u+2$. It follows from the BCH bound for cyclic codes that the minimum distance 
$d(\C_u) \geq q-2u+3$. By the Singleton bound for linear codes, we have $d(\C_u) \leq q-2u+3$. As a result, 
$d(\C_u) = q-2u+3$. Since both $\beta^i$ and $\beta^{-i}$ are roots of $g_u(x)$ for all $i$ with $u \leq i \leq q/2$, 
$\C_u$ is reversible, i.e., LCD. Clearly, $\C_u$ is a BCH code. 
\end{proof}

Reed-Solomon codes over $\gf(q)$ are MDS codes with length $q-1$. Hence, the code $\C_u$ in Theorem 
\ref{thm-sdjoin1} is not a Reed-Solomon code. In the next section, we will study the  
subfield subcodes of these MDS codes $\C_{(q+4)/4}$.

\section{The subfield subcode $(\C_{(q+4)/4}^\perp)|_{\gf(4)}$ and its dual}\label{sec-subfieldsubcode}

The subfield subcode $\C|_{\gf(r)}$ of an $[n, \ell]$ code over $\gf(r^\ell)$ is the set of codewords 
in $\C$ each of whose components is in $\gf(r)$, where $\ell$ is a positive integer and $r$ is a 
power of a prime. Our task in this section is to study the subfield subcode $(\C_{(q+4)/4}^\perp)|_{\gf(4)}$ and its dual. 
Throughout this section, let $q=2^{2h}$, where $h$ is a positive integer. 

Let notation be the same as before. Notice that the $q$-cyclotomic cosets modulo $n=q+1$ are 
$$ 
\{0\}, \{1, q\}, \ldots, \{i, q+1-i\}, \ldots, \{q/2, (q+2)/2\}. 
$$ 
For each $i$ with $0 \leq i \leq n-1$, we have 
$$ 
\m_{\beta^i}(x)=\m_{\beta^{-i}}(x), 
$$
where $\m_{\beta^i}(x)$ denotes the minimal polynomial of $\beta^i$ over $\gf(q)$. 

Recall that the check polynomial of $\C_{(q+4)/4}^\perp$ (when $u = (q+4)/4$) is 
\begin{eqnarray*}
g_{(q+4)/4}(x) 
&=& \m_{\beta^{(q+4)/4}}(x)  \cdots \m_{\beta^{q/2}}(x) \\
&=& \m_{\beta^{-(q+4)/4}}(x) \cdots \m_{\beta^{{-q/2}}}(x).  
\end{eqnarray*} 
By the Delsarte Theorem, the trace representation of $\C_{(q+4)/4}^\perp$ is 
\begin{eqnarray}\label{eqn-may251}
\C_{(q+4)/4}^\perp=
\left\{ 
\left( \sum_{i=(q+4)/4}^{q/2} \tr_{q^2/q} (a_i \beta^{ij} )       \right)_{j=0}^q: a_i \in \gf(q^2)
\right\} 
\end{eqnarray} 
and 
\begin{eqnarray*}
  %\label{eqn-may251b}
\C_{(q+4)/4} =
\left\{ 
\left( \sum_{i=0}^{q/4} \tr_{q^2/q} (a_i \beta^{ij} )       \right)_{j=0}^q: a_i \in \gf(q^2)
\right\}.  
\end{eqnarray*} 
It then follows from the Delsarte Theorem and (\ref{eqn-may251}) that 
\begin{eqnarray*}
\left( (\C_{(q+4)/4})^\perp|_{\gf(4)} \right)^\perp 
&=& \tr_{q/4} (\C_{(q+4)/4}) \\ 
&=& \left\{ 
\left( \sum_{i=0}^{q/4} \tr_{q^2/4} (a_i \beta^{ij} )       \right)_{j=0}^q: a_i \in \gf(q^2)
\right\} 
\end{eqnarray*} 
and 
\begin{eqnarray}\label{eqn-june71}
 (\C_{(q+4)/4})^\perp|_{\gf(4)}  
= \left\{ 
\left( \sum_{i=(q+4)/4}^{q/2} \tr_{q^2/4} (a_i \beta^{ij} )       \right)_{j=0}^q: a_i \in \gf(q^2)
\right\}.  
\end{eqnarray}  
It follows from the Delsarte Theorem again that the generator polynomial of the cyclic code 
$\C_{(q+4)/4}^\perp|_{\gf(4)}$ 
has zeros 
$$ 
\beta^{-0}, \beta^{-1}, \ldots, \beta^{-q/4}. 
$$ 
Consequently, the generator polynomial of $ \C_{(q+4)/4}^\perp|_{\gf(4)}$ 
has zeros 
$$ 
\beta^{0}, \beta^{1}, \ldots, \beta^{q/4}. 
$$ 

Let $C_i^{(4)}$ denote the $4$-cyclotomic coset modulo $n$ containing $i$, where $0 \leq i \leq n-1$. 
Put 
$$ 
T=\bigcup_{i=0}^{q/4} C_i^{(4)}.  
$$ 

The discussions above proved the following lemma. 

\begin{lemma}\label{lem-may271} 
Let notation be the same as before. Then 
\begin{itemize}
\item $\{0, 1, 2, \ldots, 4^{h-1}\} \subset T$ and 
\item $\{n-1, n-2, \ldots, n-4^{h-1}\} \subset T$. 
\end{itemize}
\end{lemma} 

\begin{lemma}\label{lem-may272}  
Let notation be the same as before, and 
let $4^{h-1}+1 \leq i \leq 4^{h-1}+4^{h-2}+ \cdots + 4 +1$. Then $i \in T$. 
\end{lemma} 

\begin{proof} 
The conclusion can be easily verified for $h \in \{1,2\}$. Assume now that $h \geq 3$.  
Let 
$$ 
\delta= 4^{h-1}+4^{h-2}+ \cdots + 4 +1 = \frac{4^h-1}{3}
$$ 
and 
$$ 
\gamma =4^{h-2}+4^{h-3}+\cdots + 4 = \frac{4^{h-1}-4}{3}. 
$$
For each integer $i$ with $4^{h-1}+1 \leq i \leq \delta$, there is an integer $j$ with $0 \leq j \leq \gamma$ 
such that $i=\delta - j$. 

Note that 
$$ 
2 \leq 2+3j \leq 2 + 3 \gamma = 4^{h-1}-2. 
$$ 
It then follows that there is an integer $t$ with $1 \leq t \leq h-1$ such that 
\begin{eqnarray*}
  %\label{eqn-may261}
4^{h-1}+1 \leq (2+3j)4^t \leq 4^h+1. 
\end{eqnarray*}
With this $t$, we have 
\begin{eqnarray*}
i 4^t \bmod{n} 
&=& [4^{h-1}+4^{h-2}+ \cdots + 4^t - (4^{t-1}+4^{t-2}+\cdots+1) - 4^tj ] \bmod{n}  \\ 
&=& \frac{4^h-(2+3j)4^t+1}{3} \bmod{n} \\
&=& \frac{4^h-(2+3j)4^t+1}{3}  \\
& \leq & 4^{h-1}. 
\end{eqnarray*} 
This means that $i 4^t \bmod{n} \in T$. Consequently,  $i \in T$.
\end{proof}

\begin{lemma}\label{lem-may273}
Let notation be the same as before. Then 
\begin{itemize}
\item $\{0, 1, 2, \ldots, \delta\} \subset T$ and 
\item $\{n-1, n-2, \ldots, n-\delta\} \subset T$. 
\end{itemize}
\end{lemma} 

\begin{proof}
Combining Lemmas \ref{lem-may271} and \ref{lem-may272} proves the first statement. 
Note that $a4^h \equiv n-a \pmod{n}$. The second conclusion then follows. 
\end{proof}

Define 
$$
T^c= 
\left\{\sum_{i=0}^{h-1}a_i4^i : a_0 \in \{2, 3\}, \ a_i \in \{1,2\} \mbox{ for } 1 \leq i \leq h-1 \right\}. 
$$

\begin{lemma}\label{lem-may274} 
$T^c$ is the union of some $4$-cyclotomic cosets modulo $n$. 
\end{lemma} 

\begin{proof}
Recall that $n=4^h+1$. Let $a \in T^c$ and $a=\sum_{i=0}^{h-1}a_i4^i$, where $a_0 \in \{2, 3\}, \ a_i \in \{1,2\} \mbox{ for } 1 \leq i \leq h-1$. It is easily verified that 
\begin{eqnarray*}
4a \bmod{n} = 
\left\{ 
\begin{array}{ll}
a_{h-2}4^{h-1}+ \cdots + a_14^2 + 4 + 3 & \mbox{ if } (a_0, a_{h-1})=(2,1), \\
a_{h-2}4^{h-1}+ \cdots + a_14^2 + 4 + 2 & \mbox{ if } (a_0, a_{h-1})=(2,2), \\
a_{h-2}4^{h-1}+ \cdots + a_14^2 + 2 \times 4 + 3 & \mbox{ if } (a_0, a_{h-1})=(3,1), \\
a_{h-2}4^{h-1}+ \cdots + a_14^2 + 2 \times 4 + 2 & \mbox{ if } (a_0, a_{h-1})=(3,2). 
\end{array}
\right. 
\end{eqnarray*} 
This completes the proof. 
\end{proof}

\begin{lemma}\label{lem-may275} 
$T$ and $T^c$ partition $\Z_n=\{0,1, \ldots, n-1\}$.  
\end{lemma} 

\begin{proof} 
The conclusion can be easily verified for $h \in \{1,2\}$. We consider below the case $h \geq 3$. 
By definition, the least integer in $T^c$ is 
$$ 
\min=(4^{h-1}+ \cdots + 4)+2=\frac{2^{2h}+2}{3},  
$$  
and the largest integer  in $T^c$ is  
$$ 
\max=2(4^{h-1}+ \cdots + 4)+3=\frac{2^{2h+1}+1}{3}. 
$$ 
It then follows from Lemmas \ref{lem-may273} and \ref{lem-may274} that 
$T$ and $T^c$ are disjoint. 

We now prove that every integer $a$ with $\min \leq a \leq \max$ is an element of 
either $T^c$ or $T$. Let 
$$ 
a=a_{h-1}4^{h-1}+a_{h-2}4^{h-2}+ \cdots + a_14+ a_0, 
$$ 
where $0 \leq a_i \leq 3$. Since $\min \leq a \leq \max$, we deduce that $a_{h-1} \in \{1,2\}$. 
We will carry out our discussions in the following cases. 

\subsubsection*{The case $a_0=1$} 
In this case, we have  
\begin{eqnarray*}
a4^{h-1} \bmod{n} 
&=& 4^{h-1} - (a_{h-1}4^{h-2} + \cdots + a_24+a_1) \\
&=& 4^{h-1}-\frac{a-1}{4} \\
&\leq & \frac{2 \times 4^h+1}{3 \times 4} \\
&<& \delta. 
\end{eqnarray*} 
This means that $ a4^{h-1} \bmod{n} \in T$. Consequently, $a \in T$. 

\subsubsection*{The case $a_0=0$} 
By assumption we have $a_0=0$ and $a_0 \leq \max$. Then we have 
$$
a \leq 2(4^{h-1}+4^{h-2}+ \cdots + 4).  
$$  
It then follows that 
\begin{eqnarray*}
\frac{a}{4} 
&=& a_{h-1}4^{h-2} + \cdots + a_24 + a_1 \\
&\leq& 2(4^{h-2} + \cdots + 1) \\
&=& \frac{2(4^{h-1}-1)}{3} \\
&<& 4^{h-1}. 
\end{eqnarray*}
This means that $a/4 \in T$. As a result, $a \in T$. 

\subsubsection*{The case $a_0 \in \{2,3\}$ and $a_i=0$ for some $1 \leq i \leq h-2$} 

Let $t$ be the smallest integer such that $1 \leq t \leq h-2$ and $a_t=0$. Then $a_{t-1} \neq 0$. 
In this case, we have 
\begin{eqnarray*}
a4^{h-1-t} \bmod{n} 
&=& a_t4^{h-1} + a_{t-1}4^{h-2} + \cdots + a_04^{h-1-t} -(a_{h-1}4^{h-2-t} + \cdots +a_{t+1}) \\
&<& 3(4^{h-2} + \cdots + 4^{h-1-t} ) \\ 
&=& 4^{h-1}-4^{h-1-t} \\
&<& 4^{h-1}.  
\end{eqnarray*}
This means that $ a4^{h-1-t} \bmod{n} \in T$. Consequently, $a \in T$.   

\subsubsection*{The case $a_0 \in \{2,3\}$ and $a_i \neq 0$ for all $1 \leq i \leq h-2$} 

Recall that $a_{h-1} \in \{1,2\}$. If there is no integer $i$ such that $1 \leq i \leq h-2$ and $a_i=3$. 
By the definition of $T$, we have $a \in T$. 

If there is an integer $i$ such that $1 \leq i \leq h-2$ and $a_i=3$, let $t$ be the smallest integer 
such that $1 \leq t \leq h-2$ and $a_t=3$. In this case, we have 
\begin{eqnarray*}
a4^{h-1-t} \bmod{n} 
&=& a_t4^{h-1} + a_{t-1}4^{h-2} + \cdots + a_04^{h-1-t} -(a_{h-1}4^{h-2-t} + \cdots +a_{t+1}) \\ 
& \geq & 3 \times 4^{h-1} +4^{h-2} + \cdots + 4^{h-t} + 2 \times 4^{h-1-t} - \\
& &            (2 \times 4^{h-2-t}+3 \times 4^{h-3-t} + \cdots + 3 \times 1)\\
&=& \frac{(2 \times 4^h+1) -(2\times 4^{h-1}+2-4^{h-2-t})}{3} \\
&>& \frac{2 \times 4^h+1}{3}. 
\end{eqnarray*}
This means that $ a4^{h-1-t} \bmod{n} \in T$. Consequently, $a \in T$. 

Summarizing the conclusions in the cases above, we deduce that every integer $a$ with 
$\min \leq a \leq \max$ is an element of either $T$ or $T^c$. This completes the proof of 
this lemma. 
\end{proof}

The first main result of this paper is documented in the following theorem. 

\begin{theorem}\label{thm-mainjune12}
The cyclic code $(\C_{(q+4)/4}^\perp)|_{\gf(4)}$ has parameters $[2^{2h}+1, 2^h, d^\perp]$, where 
$$ 
d^\perp \geq \frac{2(2^{2h}+2)}{3}. 
$$ 
The cyclic code $((\C_{(q+4)/4}^\perp)|_{\gf(4)})^\perp$ has parameters $[2^{2h}+1, 2^{2h}+1-2^h, d\geq 3]$. 
\end{theorem} 

\begin{proof} 
It follows from (\ref{eqn-june71}) that the generator and check polynomials of $(\C_{(q+4)/4}^\perp)|_{\gf(4)}$ are 
$$ 
g'_{(q+4)/4}(x)=\lcm\{\m'_{\beta^0}(x),  \m'_{\beta^1}(x), \ldots, \m'_{\beta^{q/4}}(x)\}
$$
and 
$$ 
h'_{(q+4)/4}(x)=\lcm\{\m'_{\beta^{(q+4)/4}}(x),  \m'_{\beta^{(q+8)/4}}(x), \ldots, \m'_{\beta^{q/2}}(x)\},
$$
respectively, where $\m'_{\beta^{i}}(x)$ denotes the minimal polynomial of $\beta^i$ over $\gf(4)$.

By definition, $|T^c|=2^h$. It then follows from Lemma \ref{lem-may275} that the degree of the check polynomial of 
 $(\C_{(q+4)/4}^\perp)|_{\gf(4)}$ is equal to $2^h$. This means that the dimension of the code 
  $(\C_{(q+4)/4}^\perp)|_{\gf(4)}$ is $2^h$.

By definition, the least integer in $T^c$ is 
$$ 
\min=(4^{h-1}+ \cdots + 4)+2=\frac{2^{2h}+2}{3},  
$$  
and the largest integer  in $T^c$ is  
$$ 
\max=2(4^{h-1}+ \cdots + 4)+3=\frac{2^{2h+1}+1}{3}. 
$$ 
It then follows that the generator polynomial of $(\C_{(q+4)/4}^\perp)|_{\gf(4)}$ has the following roots 
$$ 
\{ \beta^{-(2^{2h}-1)/3}, \ldots, \beta^{-2}, \beta^{-1}, 0, \beta^1, \beta^2, \ldots, \beta^{(2^{2h}-1)/3}\}. 
$$ 
The lower bound on $d^\perp$ follows from the BCH bound on cyclic codes. 

The dimension of $(\C_{(q+4)/4}^\perp|_{\gf(4)})^\perp$ follows from that of $(\C_{(q+4)/4}^\perp)|_{\gf(4)}$. 
Note that the defining set of the cyclic code $(\C_{(q+4)/4}^\perp|_{\gf(4)})^\perp$ is 
\begin{eqnarray}\label{eqn-may281}
\{n-i: i \in T^c\}. 
\end{eqnarray} 
By definition, 
the set in (\ref{eqn-may281}) contains the following two consecutive integers 
$$ 
n-(4^{h-1}+4^{h-2}+\cdots+4 +3), \ n-(4^{h-1}+4^{h-2}+\cdots+4 +2). 
$$
It then follows from the BCH bound that the minimum distance of $(\C_{(q+4)/4}^\perp|_{\gf(4)})^\perp$ is 
at least $3$. 
\end{proof} 

\begin{example} 
Let $h=1$. Then the cyclic code $(\C_{(q+4)/4}^\perp)|_{\gf(4)}$ has parameters $[5, 2, 4]$ and weight enumerator 
$1+15z^{5}$. The cyclic code $((\C_{(q+4)/4}^\perp)|_{\gf(4)})^\perp$ has parameters $[5, 3, 3]$.  Both codes are MDS and optimal. 
\end{example} 

\begin{example} 
Let $h=2$. Then the cyclic code $(\C_{(q+4)/4}^\perp)|_{\gf(4)}$ has parameters $[17, 4, 12]$ and weight enumerator 
$1+204z^{12}+51z^{16}$, and is distance-optimal. The cyclic code $((\C_{(q+4)/4}^\perp)|_{\gf(4)})^\perp$ has parameters $[17, 13, 4]$ 
and is distance-optimal.  
\end{example} 

\begin{example} 
Let $h=3$. Then the cyclic code $(\C_{(q+4)/4}^\perp)|_{\gf(4)}$ has parameters $[65, 8, 44]$ and weight enumerator 
$1+18720z^{44}+16380z^{48}+30240z^{52}+195z^{64}$, and is distance-optimal. The cyclic code $((\C_{(q+4)/4}^\perp)|_{\gf(4)})^\perp$ has parameters 
$[65, 57, 5]$ and is distance-optimal.  
\end{example} 

\begin{example} 
Let $h=4$. Then the cyclic code $(\C_{(q+4)/4}^\perp)|_{\gf(4)}$ has parameters $[257, 16, 172]$ and weight enumerator 
\begin{eqnarray*}
1+ 28422144z^{172}+  25794576z^{176}+ 258365184z^{180} + 234877440z^{184} + \\ 
1160570880z^{188} + 469178172z^{192}+ 1348867584z^{196}+ 301985280z^{200} + \\
394752000z^{204}+ 41942400z^{208}+ 30198528z^{212}+ 12336z^{240}+ 771z^{256}.  
\end{eqnarray*}
The cyclic code $((\C_{(q+4)/4}^\perp)|_{\gf(4)})^\perp$ has parameters $[257, 241, 8]$.  
\end{example} 

These examples indicate that the subfield subcode $\left(\C_{(q+4)/4}^\perp\right)_{\gf(4)}$ and its dual 
are distance-optimal when $h \in \{1,2,3\}$. Hence, the subfield subcode 
$\left(\C_{(q+4)/4}^\perp\right)_{\gf(4)}$ and its dual 
are very interesting in terms of the error-correcting capability.

\section{The subfield code $\C_{(q+4)/4}^{(4)}$ and its dual}

Let $\gf(r^h)$ be a finite field with $r^h$ elements, where $r$ is a power of a prime and $h$ is a positive integer. 
Given an $[n,k]$ code $\C$ over $\gf(r^h)$, we construct a new $[n, k']$ code $\C^{(r)}$ 
over $\gf(r)$ as follows. Let $G$ be a generator matrix of $\C$. Take a basis of $\gf(r^h)$ 
over $\gf(r)$. Represent each entry of $G$ as an $h \times 1$ column vector of $\gf(r)^h$ 
with respect to this basis, and replace each entry of $G$ with the corresponding $h \times 1$ column vector of $\gf(r)^h$. In this way, $G$ is modified into a $kh \times n$ matrix over 
$\gf(r)$, which generates a new code $\C^{(r)}$ over $\gf(r)$ with length $n$, which is called the subfield code of $\C$. 
By definition, the dimension $k'$ of $\C^{(r)}$ satisfies $k'\leq hk$. 
It is known that 
the subfield code $\C^{(r)}$ of $\C$ is independent of the choices of both $G$ and the 
basis of $\gf(r^h)$ over $\gf(q)$ \cite{DH19}. 

The trace representation of the subfield code  $\C^{(r)}$ of $\C$ is given in the next theorem~\cite{DH19}. 

\begin{theorem}~\cite{DH19}\label{thm-tracerepresentation}
Let $\C$ be an $[n,k]$ linear code over $\gf(r^h)$. Let $G=[g_{ij}]_{1\leq i \leq k, 1\leq j \leq n}$ be a generator matrix of $\C$. Then the trace representation of $\C^{(r)}$ is given by
\begin{eqnarray*}
\C^{(r)}= \left\{ 
\begin{array}{r}
\left(\tr_{r^h/r}\left(\sum_{i=1}^{k}a_ig_{i1}\right),
\ldots,\tr_{r^h/r}\left(\sum_{i=1}^{k}a_ig_{in}\right)\right): \\
a_1,\ldots,a_k\in \gf(r^h)
\end{array}
\right\},  
\end{eqnarray*} 
where $\tr_{r^h/r}(x)$ denotes the trace function from $\gf(r^h)$ to $\gf(r)$. 
\end{theorem}

The trace of a vector $\bc=(c_1, c_2, \ldots, c_n) \in \gf(r^h)^n$ is defined by 
$$ 
\tr_{r^h/r}(\bc)=\left(\tr_{r^h/r}(c_1),  \tr_{r^h/r}(c_2), \ldots, \tr_{r^h/r}(c_n) \right). 
$$ 
The trace code of a linear code $\C$ over $\gf(r^h)$ is defined by 
$$ 
\tr_{r^h/r}(\C)=\left\{ \tr_{r^h/r}(\bc): \bc \in \C \right\}. 
$$ 
It then follows from Theorem \ref{thm-tracerepresentation} that 
\begin{eqnarray*}
\C^{(r)} =\tr_{r^h/r}(\C). 
\end{eqnarray*} 
Hence, the subfield code $\C^{(r)}$ of a linear code $\C$ over $\gf(r^h)$ is in fact the trace code 
$\tr_{r^h/r}(\C)$ of $\C$. The following theorem then follows from the Delsarte theorem. 

\begin{theorem}\label{thm-20june91}
Let $\C$ be a linear code over $\gf(r^h)$. Then 
$$ 
\left((\C^\perp)|_{\gf(r)}\right)^\perp = \C^{(r)}. 
$$
\end{theorem}

By Theorem \ref{thm-20june91} and the definition of the subfield subcode, 
$$ 
\C|_{\gf(r)} \subseteq \left((\C^\perp)|_{\gf(r)}\right)^\perp = \C^{(r)}. 
$$
Hence, the subfield subcode $\C|_{\gf(r)}$ is a subcode of the subfield code $\C^{(r)}$. In general, 
the two codes are different.

The subfield code $\C^{(r)}$ of a code $\C$ over $\gf(r^h)$ may have good or bad parameters. 
To obtain a good  subfield code $\C^{(r)}$ over $\gf(r)$, the code  $\C$ over $\gf(r^h)$ must be 
well chosen. Several families of very good subfield codes were obtained in \cite{DH19,HD19,HDW20,WZ20}.  
The objective of this section is to present a family of very good subfield codes over $\gf(4)$ whose duals 
are also very good. Specifically, we have the following corollary, which follows from Theorem  \ref{thm-20june91}. 

\begin{corollary}\label{cor-mdssubcode}
Let notation be the same as before. We then have 
$$ 
\left(\C_{(q+4)/4}^{(4)}\right)^\perp = \left(\C_{(q+4)/4}^\perp\right)|_{\gf(4)}. 
$$
\end{corollary}

It follows from Corollary \ref{cor-mdssubcode} and the examples in Section \ref{sec-subfieldsubcode} that  
the subfield code $\C_{(q+4)/4}^{(4)}$ and its dual 
are distance-optimal when $h \in \{1,2,3\}$. This shows that the subfield code $\C_{(q+4)/4}^{(4)}$ and its dual 
are very interesting in terms of the error-correcting capability. 

We point out that $(\C_{(q+4)/4})|_{\gf(4)}$ is indeed a proper subcode of $\C_{(q+4)/4}^{(4)}$. For example, 
when $m=4$, the former has parameters $[17, 5, 9]$ and the latter has parameters $[17, 13, 4]$.   

\section{An infinite family of conjectured $3$-designs supported by the codes $\C_{(q+4)/4}^{(4)}$}\label{sec-designappl}

\subsection{The support designs of linear codes}

Let $\cP$ be a set of $v \ge 1$ elements, where $v$ is an integer, and let $\cB$ be a set of $k$-subsets of $\cP$, where $k$ is
a positive integer with $1 \leq k \leq v$. Let $t$ be a positive integer with $t \leq k$. The pair $\bD := (\cP, \cB)$ becomes an 
incidence structure when the incidence relation is the set membership. The incidence structure 
$\bD = (\cP, \cB)$ is called a $t$-$(v, k, \lambda)$ {\em design}, or simply {\em $t$-design}, if every $t$-subset of $\cP$ is contained in exactly $\lambda$ elements of
$\cB$. The elements of $\cP$ are called points, and those of $\cB$ are referred to as blocks. 
The set $\cB$ is called the block set. 
The number of blocks in $\cB$ is usually denoted by $b$. 
Let $\binom{\cP}{k}$ denote the set of all $k$-subsets of $\cP$. Then $\left (\cP, \binom{\cP}{k} \right )$ is a $k$-$(v, k, 1)$ design, 
which is called a \emph{complete design}. 
  A $t$-design is called {\em simple} if $\cB$ does not contain
any repeated blocks.
We consider only simple $t$-designs with $v > k > t$.
A $t$-$(v,k,\lambda)$ design is referred to as a
{\em Steiner system} if $t \geq 2$ and $\lambda=1$,
and is denoted by $S(t,k, v)$. By definition, the parameters of a $t$-$(v,k, \lambda )$ design have the following relation: 
\begin{align*}
\binom{v}{t} \lambda =\binom{k}{t} b.
\end{align*}

There are different approaches to constructing $t$-designs. 
A coding-theoretic construction of $t$-designs is as follows. 
Let $\C$ be a $[v, \kappa, d]$ linear code over $\gf(q)$. 
For each $k$ with $A_k \neq 0$,  let $\cB_k(\C)$ denote
the set of the supports of all codewords with Hamming weight $k$ in $\C$, where the coordinates of a codeword
are indexed by $(p_1, \ldots, p_v)$. Let $\cP(\C)=\{p_1, \ldots, p_v\}$.  The incidence structure $(\cP(\C), \cB_k(\C))$
may be a $t$-$(v, k, \lambda)$ design for some positive integer $\lambda$, which is called a
\emph{support design} of the code $\C$, and is denoted by $\bD_k(\C)$. In such a case, we say that the codewords of weight $k$ 
in $\C$ support or hold a $t$-$(v, k, \lambda)$ design, and for simplicity, we say that $\C$ supports or holds a $t$-$(v, k, \lambda)$ design. 

The following theorem, established by Assumus and Mattson and called the Assmus-Mattson Theorem, 
shows that the pair $(\cP(\C), \cB_k(\C))$ defined by 
a linear code $\C$ is a $t$-design under certain conditions \cite{AM69}.

\begin{theorem}\label{thm-designAMtheorem}
Let $\C$ be a $[v,k,d]$ code over $\gf(q)$. Let $d^\perp$ denote the minimum distance of $\C^\perp$. 
Let $w$ be the largest integer satisfying $w \leq v$ and 
$$ 
w-\left\lfloor  \frac{w+q-2}{q-1} \right\rfloor <d. 
$$ 
Define $w^\perp$ analogously using $d^\perp$. Let $(A_i)_{i=0}^v$ and $(A_i^\perp)_{i=0}^v$ denote 
the weight distribution of $\C$ and $\C^\perp$, respectively. Fix a positive integer $t$ with $t<d$, and 
let $s$ be the number of $i$ with $A_i^\perp \neq 0$ for $1 \leq i \leq v-t$. Suppose $s \leq d-t$. Then 
\begin{itemize}
\item all the codewords of weight $i$ in $\C$ support a $t$-design provided $A_i \neq 0$ and $d \leq i \leq w$, and 
\item all the codewords of weight $i$ in $\C^\perp$ support a $t$-design provided $A_i^\perp \neq 0$ and 
         $d^\perp \leq i \leq \min\{v-t, w^\perp\}$. 
\end{itemize}
\end{theorem}

The Assmus-Mattson Theorem above is a useful tool in constructing $t$-designs from linear codes 
(see, for example, \cite{Dingbook18}, \cite{Tonchev} and \cite{Tonchevhb}), but does not 
characterize all linear codes supporting $t$-designs. The reader is referred to \cite{TDX19} for a generalized 
Assmus-Mattson theorem.

Another sufficient condition for the incidence structure $(\cP(\C), \cB_k(\C))$ to be a $t$-design is via the 
automorphism group of the linear code $\C$. Before introducing this  sufficient condition, we need to recall 
several different automorphism groups of a linear code.

The set of coordinate permutations that map a code $\C$ to itself forms a group, 
which is referred to as the \emph{permutation automorphism group} of $\C$
and denoted by $\PAut(\C)$. If the length of $\C$ is $n$, then $\PAut(\C)$ is a subgroup of the
\emph{symmetric group} $\Sym_n$.

A \emph{monomial matrix} over $\gf(q)$ is a square matrix that has exactly one
nonzero element of $\gf(q)$  in each row and column. It is easily seen that a monomial matrix $M$ can be written either in
the form $DP$ or the form $PD_1$, where $D$ and $D_1$ are both diagonal matrices and $P$ is a permutation
matrix. Clearly, the set of monomial matrices that map $\C$ to itself forms a group denoted by $\MAut(\C)$,  which is called the
\emph{monomial automorphism group\index{monomial automorphism group}} of $\C$. By definition, 
we clearly have 
$$
\PAut(\C) \subseteq \MAut(\C).
$$

The \textit{automorphism group}\index{automorphism group} of $\C$, denoted by $\GAut(\C)$, is the set
of maps of the form $M\gamma$,
where $M$ is a monomial matrix and $\gamma$ is a field automorphism, that map $\C$ to itself. In the binary
case, $\PAut(\C)$,  $\MAut(\C)$ and $\GAut(\C)$ are the same. If $q$ is a prime, $\MAut(\C)$ and
$\GAut(\C)$ are identical. According to their definitions, we have in general the following relations: 
$$
\PAut(\C) \subseteq \MAut(\C) \subseteq \GAut(\C).
$$

By definition, every element in $\GAut(\C)$ is of the form $DP\gamma$, where $D$ is a diagonal matrix,
$P$ is a permutation matrix, and $\gamma$ is an automorphism of $\gf(q)$.
The automorphism group $\GAut(\C)$ is said to be {\em $t$-transitive} if for every pair of $t$-element ordered
sets of coordinates, there is an element $DP\gamma$ of the automorphism group $\GAut(\C)$ such that its
permutation part $P$ sends the first set to the second set. The automorphism group $\GAut(\C)$ is said to be {\em $t$-homogeneous} if for every pair of $t$-element 
sets of coordinates, there is an element $DP\gamma$ of the automorphism group $\GAut(\C)$ such that its
permutation part $P$ sends the first set to the second set.

Using the automorphism group of a linear code $\C$, the following theorem gives another sufficient condition 
for the code $\C$ to hold $t$-designs.

\begin{theorem}~\cite[p. 308]{HP03}\label{thm-designCodeAutm}
Let $\C$ be a linear code of length $n$ over $\gf(q)$ such that $\GAut(\C)$ is $t$-transitive 
or $t$-homogeneous. Then the codewords of any weight $i \geq t$ of $\C$ hold a $t$-design.
\end{theorem}

\subsection{The support $3$-designs of $\C_{(q+4)/4}^{(4)}$ and its dual}

Corollary \ref{cor-mdssubcode} gives that 
$$ 
\left(\C_{(q+4)/4}^{(4)}\right)^\perp = \left(\C_{(q+4)/4}^\perp\right)|_{\gf(4)}. 
$$ 
We are now ready to present families of $3$-designs supported by the codes $(\C_{(q+4)/4})^\perp|_{\gf(4)}$ and its dual 
$\C_{(q+4)/4}^{(4)}$. The second main result of this paper is the following. 

\begin{theorem}\label{thm-june11} 
Let notation be the same as before. For each $k$ with $1 \leq k <q$ and $A_k^\perp \neq 0$, the 
incidence structure $(\cP, \cB_k((\C_{(q+4)/4})^\perp|_{\gf(4)}))$ is a $3$-$(q+1, k, \lambda)$ 
design, where $\lambda$ is a positive integer and $A_k^\perp$ is the number of codewords with 
weight $k$ in $(\C_{(q+4)/4})^\perp|_{\gf(4)}$. 

For each $k$ with $1 \leq k <q$ and $A_k \neq 0$, the 
incidence structure $(\cP, \cB_k(\C_{(q+4)/4}^{(4)} )$ is a $3$-$(q+1, k, \lambda)$ 
design, where $\lambda$ is a positive integer and $A_k$ is the number of codewords with 
weight $k$ in $\C_{(q+4)/4}^{(4)}$.  
\end{theorem}

\begin{example} 
Let $h=2$. Then the cyclic code $(\C_{(q+4)/4}^\perp)|_{\gf(4)}$ has parameters $[17, 4, 12]$ and weight enumerator 
$1+204z^{12}+51z^{16}$, and is distance-optimal. The cyclic code $((\C_{(q+4)/4}^\perp)|_{\gf(4)})^\perp$ has parameters $[17, 13, 4]$ 
and is distance-optimal.  Further, the minimum weight codewords in $(\C_{(q+4)/4}^\perp)|_{\gf(4)}$ support a $3$-$(17, 12, 22)$ design.  
\end{example} 

\begin{example}\label{exam-h=3} 
Let $h=3$. Then the cyclic code $(\C_{(q+4)/4}^\perp)|_{\gf(4)}$ has parameters $[65, 8, 44]$ and weight enumerator 
$1+18720z^{44}+16380z^{48}+30240z^{52}+195z^{64}$, and is distance-optimal. The cyclic code $((\C_{(q+4)/4}^\perp)|_{\gf(4)})^\perp$ has parameters 
$[65, 57, 5]$ and is distance-optimal.  
Further, the minimum weight codewords in $(\C_{(q+4)/4}^\perp)|_{\gf(4)}$ support a $3$-$(65, 44, 1892)$ design.  
\end{example} 

Example \ref{exam-h=3} shows that the Assmus-Mattson theorem cannot be used to prove the 
$3$-design property of the incidence structure $(\cP, \cB_k((\C_{(q+4)/4})^\perp|_{\gf(4)}))$. 
Since $q$ is even,  $\PGL_2(\gf(q))=\PSL_2(\gf(q))$. According to a theorem proved in \cite{DTT2020}, 
the permutation automorphism 
group of the code $(\C_{(q+4)/4})^\perp|_{\gf(4)}$ is smaller than $\PGL_2(\gf(q))$.  It looks hard to use Theorem \ref{thm-designCodeAutm} 
to prove the 
$3$-design property of the incidence structure $(\cP, \cB_k((\C_{(q+4)/4})^\perp|_{\gf(4)}))$ 
either.  In the next subsection, we will use a newly developed approach in \cite{DTT2020} to prove Theorem \ref{thm-june11}.   

\subsection{The proof of Theorem \ref{thm-june11}} 

\subsubsection{Group actions and $t$-designs}\label{sec-august11td} 

A \emph{permutation group} is a subgroup of the \emph{symmetric group} $\mathrm{Sym}(X)$, where $X$ is a finite set.
More generally, an \emph{action} $\sigma$ of a finite group $G$ on a set $X$ is a homomorphism $\sigma$ from $G$ to $\mathrm{Sym}(X)$.
We denote the image $\sigma(g)(x)$ of $x\in X$ under $g\in G$ by $g(x)$ when no confusion can arise.
The \emph{$G$-orbit} of $x\in X$ is $\mathrm{Orb}_{x}=\{g(x): g \in G\}$.
The \emph{stabilizer} of $x$ is $\mathrm{Stab}_x=\{g \in G: g(x)=x\}$. The length of the orbit of $x$ is given by
\[ \left | \mathrm{Orb}_{x} \right | = \left | G \right | / \left | \mathrm{Stab}_x \right |.\]

One criterion to measure the level of symmetry is the \emph{degree} of \emph{transitivity} and \emph{homogeneousity} of the group. 
Recall that a group $G$ acting on a set $X$ is called \emph{$t$-transitive} (resp., \emph{$t$-homogeneous}) if
for any two $t$-tuples $(x_1, \cdots, x_t), (x_1', \cdots, x_t')$
of pairwise distinct elements from $X$ (resp., two $t$-subsets $\{x_1, \cdots, x_t\}$, $\{x_1', \cdots, x_t'\}$ of $X$) there is some $g\in G$ such that $\left (x_1', \cdots, x_t' \right )= \left (g(x_1), \cdots, g(x_t) \right )$
(resp., $\left \{x_1', \cdots, x_t' \right \}=\left \{g(x_1), \cdots, g(x_t) \right \}$).

We recall a well-known general fact that (see \cite[Proposition 4.6]{BJL}), for a $t$-homogeneous group $G$ on a finite set $X$ with
$ |X|=\nu$  and a subset $B$ of $X$ with $|B| =k >t$, the pair $(X, \mathrm{Orb}_{B})$ is a $t$-$(\nu, k , \lambda)$ design,
where $\mathrm{Orb}_{B}$ is the set of images of $B$ under
the group $G$, $\lambda=\frac{\binom{k}{t} |G| }{\binom{\nu}{t} | \mathrm{Stab}_{B}|}$ and $\mathrm{Stab}_{B}$ is the setwise stabilizer of
$B$ in $X$. Let $\binom{X}{k}$ be the set of subsets of $X$ consisting of $k$ elements.
A nonempty subset $\mathcal B$ of $\binom{X}{k}$ is called \emph{invariant} under
 $G$ if $\mathrm{Orb}_{B} \subseteq \mathcal B$ for any $B \in \mathcal B$.
If this is the case, it means that the pair $(X, \mathcal B)$ is a $t$-$(\nu, k , \lambda)$ design
admitting $G$ as an automorphism group for some $\lambda$. 

\subsubsection{The projective general linear group $\PGL_2(\gf(q))$} 

The \emph{projective linear group $\PGL_2(\gf(q))$ of degree two} is defined as the group of invertible $2\times 2$ matrices with entries in $\gf(q)$,
 modulo the scalar matrices, $\begin{bmatrix}
a & 0\\
0 &  a
\end{bmatrix}$, where $a\in \gf(q)^*$.
Note that the group $\PGL_2(\gf(q))$ is generated by the matrices
$\begin{bmatrix}
a & 0\\
0 &  1
\end{bmatrix}$,
$\begin{bmatrix}
1 & b\\
0 &  1
\end{bmatrix}$
and
$\begin{bmatrix}
0 & 1\\
1 &  0
\end{bmatrix}$, where $a\in \gf(q)^*$ and $b\in \gf(q)$.

Here the following convention for the action of $\PGL_2(\gf(q))$ on the projective line $\mathrm{PG}(1,\gf(q))$
is used. A matrix $\begin{bmatrix}
a & b\\
c &  d
\end{bmatrix}
\in \PGL_2(\gf(q))$ acts on $\mathrm{PG}(1,\gf(q))$ by
\begin{eqnarray}\label{eq:action-PGL(2,q)}
\begin{array}{c}
(x_0 : x_1) \mapsto \begin{bmatrix}
a & b\\
c &  d
\end{bmatrix} (x_0 : x_1) =   (a x_0 +b x_1 : c x_0 +d x_1),
\end{array}
\end{eqnarray}
or, via the usual identification of $\gf(q) \cup \{ \infty\}$ with  $\mathrm{PG}(1,\gf(q))$, by linear fractional transformation
\begin{eqnarray}\label{eq:action-infty}
\begin{array}{c}
 x \mapsto \frac{a x +b  }{c x +d}.
\end{array}
\end{eqnarray}
This is an action on the left, i.e., for $\pi_1, \pi_2 \in  \PGL_2(\gf(q))$
and $x \in  \PG(1,\gf(q))$ the following holds: $\pi_1 (\pi_2(x)) = (\pi_1 \pi_2)(x)$.
The action of $\PGL_2(\gf(q))$ on $\PG(1,\gf(q))$ defined in (\ref{eq:action-infty}) is sharply $3$-transitive, i.e.,
for any distinct $a, b, c \in \gf(q) \cup \{\infty\}$ there is $\pi \in \PGL_2(\gf(q))$
taking $\infty$ to $a$, $0$ to $b$,  and $1$ to $c$.
 In fact, $\pi$ is uniquely determined and it equals
\[ \pi= \begin{bmatrix}
 a(b-c) & b(c-a)\\
 b-c & c-a
 \end{bmatrix}.
\]
Thus, $\PGL_2(\gf(q))$ is in one-to-one correspondence with the set of ordered triples $(a,b,c)$ of distinct elements in $\gf(q) \cup \{\infty\}$, and in particular
\begin{eqnarray}\label{eq:cardinality of PGL(2,q)}
\begin{array}{c}
| \PGL_2(\gf(q))| = (q+1)q(q-1).
\end{array}
\end{eqnarray}

Two subgroups $H_1$ and $H_2$ of a group $G$ are said to be \emph{conjugate} if there is a $g \in G$ such 
that $gH_1g^{-1}=H_2$. It is easily seen that this conjugate relation is an equivalence relation on the set of all subgroups 
of $G$, and is called the conjugacy.

\subsubsection{Another representation of the action of $\PGL_2(\gf(2^m))$ on the projective line $\PG(1,\gf(2^m))$}

To prove Theorem \ref{thm-june11}, 
we give another representation of the action of $\PGL_2(\gf(2^m))$ on the projective line $\PG(1,\gf(2^m))$.
Let $U_{q+1}$  be the subset of the projective line $\PG(1,\gf(q^2))=\gf(q^2) \cup \{\infty\}$ consisting of all the $(q+1)$-th roots of unity. Denote by $\mathrm{Stab}_{U_{q+1}}$ the setwise stabilizer of $U_{q+1}$ under the action of $\PGL_2(\gf(q^2))$ on $\PG(1,\gf(q^2))$.

\begin{proposition}\label{prop:three-types} \cite{DTT2020}
Let $q=2^m$. Then the setwise stabilizer $\mathrm{Stab}_{U_{q+1}}$ of $U_{q+1}$ consists of
the following three types of linear fractional transformations:
\begin{enumerate}
\item[(I)] $u \mapsto u_0 u$, where $u_0 \in U_{q+1}$;
\item[(II)] $u \mapsto u_0 u^{-1}$, where $u_0 \in U_{q+1}$;
\item[(III)] $u \mapsto \frac{u+c^q u_0}{c u +u_0}$, where $u_0 \in U_{q+1}$ and $c \in \gf(q^2)^* \setminus U_{q+1}$.
\end{enumerate}
\end{proposition}

The following result follows from Proposition \ref{prop:three-types} directly.

\begin{corollary}\label{cor:generators-three-types}  \cite{DTT2020} 
Let $q=2^m$. Then the setwise stabilizer $\mathrm{Stab}_{U_{q+1}}$ of $U_{q+1}$ is generated by
the following three types of linear fractional transformations:
\begin{enumerate}
\item[(I)] $u \mapsto u_0 u$, where $u_0 \in U_{q+1}$;
\item[(II)] $u \mapsto u^{-1}$;
\item[(III)] $u \mapsto \frac{u+c^q }{c u +1}$, where  $c \in \gf(q^2)^* \setminus U_{q+1}$.
\end{enumerate}
\end{corollary} 

The following proposition shows that the action of  $\mathrm{Stab}_{U_{q+1}}$ on $U_{q+1}$
and the action of $\PGL_2(\gf(2^m))$ on $\PG(1,\gf(2^m))$ are equivalent.

\begin{proposition}\label{prop:Stab-PGL}  \cite{DTT2020} 
Let $q=2^m$ and $\mathrm{Stab}_{U_{q+1}}$   the setwise stabilizer of $U_{q+1}$.
Then $\mathrm{Stab}_{U_{q+1}}$  is conjugate in $\PGL_2(\gf(2^{2m}))$ to the group $\PGL_2(\gf(2^m))$,
and its action on $U_{q+1}$ is equivalent to the action of $\PGL_2(\gf(2^m))$ on $\PG(1,\gf(2^m))$.
\end{proposition} 

\subsubsection{The proof of Theorem \ref{thm-june11}} 

Let $q=4^h$ and $n=q+1$. Define 
\begin{eqnarray}\label{eqn-August12201}
E=\left \{ 1+\sum_{i=0}^{h-1} e_i 4^i : e_i = 1 \text{ or } 2 \right \} = \{n-i: i \in T^c\}=T^c,
\end{eqnarray} 
where 
$$
T^c= 
\left\{\sum_{i=0}^{h-1}a_i4^i : a_0 \in \{2, 3\}, \ a_i \in \{1,2\} \mbox{ for } 1 \leq i \leq h-1 \right\}. 
$$
By the proof of Theorem \ref{thm-mainjune12}, $E$ is the defining set of 
$$
\C_{(q+4)/4}^{(4)} = \left( \left(\C_{(q+4)/4}^\perp  \right)|_{\gf(4)} \right)^\perp.  
$$ 
We index the coordinates in the codewords in both codes with the elements in $U_{q+1}$. When 
an element $\sigma \in \mathrm{Stab}_{U_{q+1}}$ acts on a codeword in $\C_{(q+4)/4^{(4)}}$, it acts 
on the coordinates of the codeword, and thus on the set of supports of all nonzero codewords in 
 $\C_{(q+4)/4}^{(4)}$. 

Let $\cB(\C_{(q+4)/4}^{(4)})$ denote the set of all supports of all nonzero codewords in $\C_{(q+4)/4}^{(4)}$. 
Recall that $\C_{(q+4)/4}^{(4)}$ is a cyclic code. It follows that $\cB(\C_{(q+4)/4}^{(4)})$ 
is fixed by any element $\sigma(u)=u_0u$ of  $\mathrm{Stab}_{U_{q+1}}$, where $u_0 \in U_{q+1}$.  

By (\ref{eqn-August12201}) we deduce that $i \in E$ if and only if $n-i \in E$. It then follows that $\cB(\C_{(q+4)/4}^{(4)})$ 
is fixed by the element $\sigma(u)=u^{-1}$ of  $\mathrm{Stab}_{U_{q+1}}$.  

We finally prove that the following elements of $\mathrm{Stab}_{U_{q+1}}$ 
\begin{eqnarray}\label{eqn-august11group}
u \mapsto \frac{u+c^q}{cu+1}, \ \ c \in \gf(q^2)^* \setminus U_{q+1} 
\end{eqnarray} 
fix $\cB(\C_{(q+4)/4}^{(4)})$. To this end, we need to prove a few lemmas below. 

\begin{lemma}\label{lem:x(x+1)--q^2}
Let $w=\sum_{i=0}^{2h-1} 2 \cdot 4^i$, $e=1+\sum_{i=0}^{h-1} e_i 4^i$ and $f=x^{q^2-1-qe}(x+1)^{w+(q-1)e-q^2+1} \in \gf(2)[x]$, where $e_i \in \{1,2\}$.
Then $$f= \sum \limits_{\begin{array}{c}  0\le v_{i} \le 2- e_i \\ 0\le v_{h+i} \le e_i-1 \end{array}} x^{ \sum_{i=0}^{h-1}  v_i  4^i + \sum_{i=0}^{h-1} (3-e_i+ v_{i+h})4^{h+i} -1}.$$
\end{lemma}

\begin{proof}
It is easy to check that
\begin{eqnarray}\label{eq:x^(q^2-1)}
\begin{array}{rl}
q^2-1-qe &= \sum_{i=0}^{2h-1} 3 \cdot 4^i -\sum_{i=0}^{h-1} e_i 4^{h+i}-4^h\\
&=\sum_{i=0}^{h-1} 3 \cdot 4^i +\sum_{i=0}^{h-1} (3-e_i) 4^{h+i} -4^h\\
&=\sum_{i=0}^{h-1} (3-e_i) 4^{h+i}-1.
\end{array}
\end{eqnarray}

An easy computation yields that
\begin{eqnarray}\label{eq:q^2-1}
\begin{array}{l}
w+(q-1)e-q^2+1\\
=w+(q-1) \left (1+\sum_{i=0}^{h-1} e_i 4^i \right )-q^2+1\\
=(q-1)\sum_{i=0}^{h-1} e_i 4^i+w-(q^2-1)+(q-1)\\
=-\sum_{i=0}^{h-1} e_i 4^i+ \sum_{i=0}^{h-1} e_i 4^{h+i}+\sum_{i=0}^{h-1}3 \cdot 4^{i} - \sum_{i=0}^{2h-1} 4^i\\
=\sum_{i=0}^{h-1} (2-e_i) 4^i+ \sum_{i=0}^{h-1} (e_i-1) 4^{h+i}\\
=\sum_{i=0}^{2h-1} e_i' 4^i,
\end{array}
\end{eqnarray}
where $e_i'=2-e_i$ for $0 \le i \le h-1$ and $e_i'=e_{i-h}-1$ for $h \le i \le 2h-1$.
In particular, $e_i'\in \left \{ 0,1 \right \}$.
A straightforward computation from (\ref{eq:q^2-1})  shows  that
\begin{eqnarray}\label{eq:v'i}
\begin{array}{rl}
f&=x^{q^2-1-qe} (x+1)^{\sum_{i=0}^{h-1} (2-e_i) 4^i+ \sum_{i=0}^{h-1} (e_i-1) 4^{h+i}}\\
&=x^{q^2-1-qe} \prod_{i=0}^{2h-1}(x^{4^i}+1)^{e_i'}\\
&= x^{q^2-1-qe} \sum \limits_{\begin{array}{c} v_1, \cdots, v_{2h-1}\in \{0,1\}\end{array}}\left (  \prod_{i=0}^{2h-1} \binom{e_i'}{v_i} \right ) x^{\sum_{i=0}^{2h-1} v_i 4^i}\\
&= x^{q^2-1-qe} \sum \limits_{\begin{array}{c}  0\le v_{i}\le e_i'\end{array}} x^{\sum_{i=0}^{2h-1} v_i 4^i}.
\end{array}
\end{eqnarray}
Combining (\ref{eq:v'i}) with (\ref{eq:x^(q^2-1)}) gives that
\begin{eqnarray}
\begin{array}{rl}
f&= \sum \limits_{\begin{array}{c}  0\le v_{i}\le e_i'\\ 0\le v_{h+i}\le e_{h+i}'\end{array}} x^{ \sum_{i=0}^{h-1}  v_i 4^i + \sum_{i=0}^{h-1} (3-e_i+ v_{i+h})4^{h+i} -1}.
\end{array}
\end{eqnarray}
This completes the proof.
\end{proof}

\begin{lemma}\label{lem:x(x+1)--q}
Let $e_i \in \{1,2\}$ for $0 \le i \le h-1$. Let $\ell = \sum_{i=0}^{h-1}  v_i 4^i + \sum_{i=0}^{h-1} (3-e_i+ v_{i+h})4^{h+i} -1 $ where
$0\le v_i\le 2-e_i$ and $0\le v_{i+h} \le e_i-1$. Then the remainder of $\ell$  divided by $q+1$ belongs to the set $E$. 
\end{lemma}
\begin{proof}
A standard computation shows that 
\begin{eqnarray}\label{eq:ell--q+1}
\begin{array}{rl}
\ell & \equiv \sum_{i=0}^{h-1}  v_i 4^i - \sum_{i=0}^{h-1} (3-e_i+ v_{i+h})4^{i} -1 \pmod{q+1}\\
& \equiv \sum_{i=0}^{h-1} (e_i+ v_i-v_{i+h})4^{i} -4^h \pmod{q+1}\\
& \equiv 1+ \sum_{i=0}^{h-1} (e_i+ v_i-v_{i+h})4^{i} \pmod{q+1}.
\end{array}
\end{eqnarray}
It is a simple matter to check that 
\begin{eqnarray}\label{i--i+h}
(v_i, v_{h+i})=\left \{
\begin{array}{cl}
(0,0) \text{ or } (1,0), & \text{ if } e_i=1,\\
(0,0) \text{ or } (0,1), & \text{ if } e_i=2.
\end{array}
\right .
\end{eqnarray}
Plugging (\ref{i--i+h}) into (\ref{eq:ell--q+1}) yields the desired conclusion. 

\end{proof}

\begin{lemma}\label{lem:u+c^q--cu+1}
Let $w=\sum_{i=0}^{2h-1} 2 \cdot 4^i$ and $e \in E$. Let $c\in \gf(q^2)^* \setminus U_{q+1}$. Then
there exists an array $\left ( a_{\ell} \right )_{\ell \in E} \in \gf(q^2)^{2^h}$ such that
\begin{eqnarray*}
\begin{array}{c}
\left ( u+c^q \right )^e \left ( cu+1 \right )^{w-e}= \sum_{\ell \in E} a_{\ell} u^{\ell}, \text{ for any } u \in U_{q+1}.
\end{array}
\end{eqnarray*}
\end{lemma}
\begin{proof}
It is clear that $u+c^q=(u^{-1}+c)^q=u^{-q}(cu+1)^q$.
An immediate verification gives that
\begin{eqnarray}\label{eq:u--cu}
\begin{array}{l}
\left ( u+c^q \right )^e \left ( cu+1 \right )^{w-e}\\
=u^{-qe}(cu+1)^{w+(q-1)e}\\
=u^{q^2-1-qe}(cu+1)^{w+(q-1)e-q^2+1}\\
=c^{qe}(cu)^{q^2-1-qe}(cu+1)^{w+(q-1)e-q^2+1}.
\end{array}
\end{eqnarray}
Applying Lemmas \ref{lem:x(x+1)--q^2} and  \ref{lem:x(x+1)--q} to  (\ref{eq:u--cu}) yields the desired conclusion. 
\end{proof}

The following very useful result is derived immediately from Lemma \ref{lem:u+c^q--cu+1}.

\begin{lemma}\label{lem:mod-trans-f}
Let $w=\sum_{i=0}^{2h-1} 2 \cdot 4^i$. Let $c\in \gf(q^2)^* \setminus U_{q+1}$. 
and $f(u)$ a function from $U_{q+1}$ to $\gf(q^2)$ given by $u \mapsto \sum_{e\in E} a_e u^e$.
Then there exists an array $\left ( b_{\ell} \right )_{\ell \in E} \in \gf(q^2)^{2^h}$ such that
\[(cu+1)^{w} f\left (  \frac{u+c^q}{cu+1}\right )= \sum_{\ell \in E} b_{\ell}  u^{\ell}.\]

\end{lemma}

It then follows from Lemma \ref{lem:mod-trans-f} that permutation described in (\ref{eqn-august11group}) fixes 
$\cB(\C_{(q+4)/4}^{(4)})$. Combining the discussions above and Corollary \ref{cor:generators-three-types}, 
we deduce that  $\cB(\C_{(q+4)/4}^{(4)})$ is fixed by $\mathrm{Stab}_{U_{q+1}}$. Hence, the same is true 
for  $\cB((\C_{(q+4)/4}^{(4)})^\perp)$.  Recall that $\mathrm{Stab}_{U_{q+1}}$ acts triply on $U_{q+1}$. The 
$3$-design property of the designs held in both codes follows from the discussions in Section \ref{sec-august11td}. 
This completes the proof of  Theorem~\ref{thm-june11}.

\section{Concluding remarks} 

The contributions of this paper are the family of quaternary cyclic codes $(\C_{(q+4)/4}^\perp)|_{\gf(4)}$ 
and their duals documented in Theorem \ref{thm-mainjune12} and the two infinite families of $3$-designs 
supported by these codes in Theorem~\ref{thm-june11}. 
The codes $(\C_{(q+4)/4}^\perp)|_{\gf(4)}$ would be very interesting due to the following: 
\begin{itemize} 
\item They are reversible BCH codes obtained from a family of MDS codes. 
\item They are distance-optimal when $m \in \{2,4,6\}$ and very good in general.  
\item They may be the first family of quaternary codes supporting an infinite family of $3$-designs. 
\end{itemize}

Naturally, we have a family of quaternary cyclic codes $(\C_{u}^\perp)|_{\gf(4)}$ for each $u$ with $2 \leq u \leq q/2$. 
Among all the $q/2$ families of quaternary cyclic codes $(\C_{u}^\perp)|_{\gf(4)}$, the family of codes $(\C_{(q+4)/4}^\perp)|_{\gf(4)}$ is very special in the sense that they have the best parameters and support $3$-designs, according to our experimental data. 

The results of this paper demonstrate that the subfield codes of some MDS  codes could be very interesting. It would be worthy to studying the subfield codes of Reed-Solomon codes.

%\section*{Acknowledgments}


\begin{thebibliography}{99} 

%\bibitem{AK92} 
%E. F. Assmus, Jr, J. D. Key, \textit{Designs and Their Codes}, Cambridge Tracts in Mathematics, 
%vol. 103, Cambridge University Press, Cambridge, 1992.   



%\bibitem{MagmaHK} 
%J. Cannon, W. Bosma, C. Fieker, E. Stell, \emph{Handbook of Magma Functions,} 
%Version 2.19, Sydney, 2013 

\bibitem{AM69} E. F. Assmus Jr., H. F. Mattson Jr., ``New 5-designs," \emph{J. Comb. Theory}, 
vol. 6, no. 2,  pp. 122--151, March 1969. 

\bibitem{BJL}
T. Beth, D. Jungnickel, H. Lenz, Design Theory, Cambridge University
Press, Cambridge, 1999. 


\bibitem{Dingbook15}
C. Ding, \emph{Codes from Difference Sets}, World Scientific, Singapore, 2015.

\bibitem{Dingbook18}
C. Ding, \emph{Designs from Linear Codes}, World Scientific, Singapore, 2018.

\bibitem{DH19} 
C. Ding, Z. Heng, ``The subfield codes of ovoid codes," \emph{IEEE Trans. Inf. Theory}, vol. 65, no. 8, pp. 4715--4729, August 2019.


\bibitem{DingTang19} 
C. Ding, C. Tang, ``Infinite families of near MDS codes holding $t$-designs,"   
\emph{IEEE Trans. Inf. Theory}, vol. 66, no. 9, pp. 5419--5428, September 2020. 

\bibitem{DTT2020} 
C. Ding, C. Tang, V. D. Tonchev, ``The projective general linear groups $\PGL_2(\gf(2^m))$ and linear codes of length 
$2^m+1$, \emph{Des. Codes Cryptogr.}, 2021,  https://doi.org/10.1007/s10623-021-00888-2

\bibitem{HD19} 
Z. Heng, C. Ding, ``The subfield codes of hyperoval and conic codes," \emph{Finite Fields and Their Applications}, 
vol. 56, pp. 308--331, 2019. 

\bibitem{HDW20} 
Z. Heng, C. Ding, W. Wang, ``Optimal binary linear codes from maximal arcs," \emph{IEEE Trans. Inf. Theory}, 
vol. 66, no. 9, pp. 5387--5394, September 2020.

\bibitem{HP03} 
W. C. Huffman, V. Pless, \emph{Fundamentals of
Error-Correcting Codes}, Cambridge University Press, Cambridge, 2003. 

\bibitem{LWL19}
C. Li, P. Wu, F. Liu, 
``On two classes of primitive BCH codes and some related codes,"  
\emph{IEEE Trans. Inform. Theory}, vol. 65, no. 6, pp. 3830--3840, June 2009. 

\bibitem{LiSIAM}
S. Li,  ``The minimum distance of some narrow-sense primitive BCH codes,"  
\emph{SIAM J. Discrete Math.}, vol. 31, no. 4, pp. 2530--2569, April 2017.  

\bibitem{LLFLR} 
Y. Liu, R. Li, Q. Fu, L. Lu, Y. Rao, 
``Some binary BCH codes with length $n=2^m+1$,"  
\emph{Finite Fields and Their Applications}, vol. 55,  pp. 109--133, 2019. 
  

\bibitem{MS77}
F. J. MacWilliams, N. J. A. Sloane, \emph{The Theory of Error-Correcting Codes},
North-Holland, Amsterdam,  1977.

\bibitem{SYW} 
X. Shi, Q. Yue, Y. Wu, 
``The dual-containing primitive BCH codes with the maximum designed distance and their applications to quantum codes,"  
\emph{Des. Codes Cryptogr.}, vol. 87, no. 9, pp. 2165--2183, 2019. 

\bibitem{TangDing20} 
C. Tang, C. Ding,  ``An infinite family of linear codes supporting 4-designs," \emph{IEEE Trans. Inf. Theory}, 
vol. 67, no. 1, pp. 244--254, Jan. 2021.  

\bibitem{TDX19} 
C. Tang, C. Ding, M. Xiong, ``Codes, differentially $\delta$-uniform functions and $t$-designs,"  
\emph{IEEE Trans. Inf. Theory}, vol. 66, no. 6, pp. 3691--3703, June 2020.   

\bibitem{Tonchev}
V. D. Tonchev, ``Codes and designs,"  in:
Handbook of Coding Theory, Vol. II, V. S. Pless, and W. C. Huffman, (Editors), Elsevier, Amsterdam, 1998, pp. 1229--1268.


\bibitem{Tonchevhb}
V. D. Tonchev, ``Codes," in:
Handbook of Combinatorial Designs, 2nd Edition, C. J. Colbourn, and J. H. Dinitz, (Editors), CRC Press, New York, 2007, pp.677--701. 


\bibitem{WZ20} 
X. Wang, D. Zheng, ``The subfield codes of several classes of linear codes," 
\emph{Cryptography and Communications}, vol. 12, pp. 1111--1131, April 2020. 

\bibitem{YLLY} 
H. Yan, H. Liu, C. Li, S. Yang, 
``Parameters of LCD BCH codes with two lengths," \emph{Adv. in Math. of Comm.}, vol. 12, no. 3, pp. 579--594, 2018. 


\end{thebibliography}
\end{document}